\documentclass{article}
\usepackage{amsmath,amssymb,amsthm}
\usepackage{authblk}
\usepackage{color}
\usepackage{graphicx}
\usepackage{float}
\usepackage{hyperref}
\usepackage{fancyhdr}
\hypersetup{colorlinks=true,linkcolor=black,citecolor=black}
 \usepackage[toc,page]{appendix}
\usepackage{lipsum}
\usepackage{algorithm}
\usepackage{multirow}
\usepackage[noend]{algpseudocode}
\usepackage{booktabs}
\usepackage{graphicx}
\usepackage{comment}
\usepackage[left=4cm,right=4cm,top=2.5cm,bottom=2.5cm]{geometry}

\newtheorem{theorem}{Theorem}
\newtheorem{proposition}[theorem]{Proposition}

\theoremstyle{definition}
\newtheorem{definition}[theorem]{Definition}
\newtheorem{example}{Example}
\theoremstyle{remark}
\newtheorem{remark}{Remark}

\graphicspath{{figs/}}

\usepackage[all]{xy}           
\usepackage{subcaption}
\usepackage{caption}
\newcount\ancho \newcount\anchom \newcount\anchoa
\newcount\anchob \newcount\altura

\newcommand{\ltilde}[3][0]{\altura=0 \advance\altura by #1
	\ancho=#2 \anchom=\ancho \divide\anchom by 2
	\anchoa=\ancho \divide\anchoa by 4
	\anchob=\anchom \advance\anchob by \anchoa
	\kern-3pt \begin{array}[b]{c}
		\begin{picture}(1,1)(\anchom,-\altura)
		\qbezier(0,2)(\anchoa,5)(\anchom,2)
		\qbezier(\anchom,2)(\anchob,-1)(\ancho,4)
		\qbezier(0,2)(\anchoa,4.5)(\anchom,1.8)
		\qbezier(\anchom,1.8)(\anchob,-1.5)(\ancho,4)
		\end{picture} \\[-4pt]{#3}
	\end{array} \kern-4pt    }

\newcommand{\lhat}[3][0]{\altura=0 \advance\altura by #1
	\ancho=#2 \anchom=\ancho \divide\anchom by 2
	\anchoa=\ancho \divide\anchoa by 4
	\anchob=\anchom \advance\anchob by \anchoa
	\kern-3pt \begin{array}[b]{c}
		\begin{picture}(1,1)(\anchom,-\altura)
		\qbezier(0,2)(\anchoa,4)(\anchom,6)
		\qbezier(\anchom,6)(\anchob,4)(\ancho,2)
		\qbezier(0,2)(\anchoa,3.8)(\anchom,5.6)
		\qbezier(\anchom,5.6)(\anchob,3.8)(\ancho,2)
		\end{picture} \\[-4pt] {#3}
	\end{array} \kern-4pt    }

\makeatletter
\newcommand\prol{\@ifstar{\@proldf}{\@prolpf}}  
\def\@prolpf{\@ifnextchar[{\@prolpf@wrt}{\@prolpf@}}
\def\@prolpf@wrt[#1]#2{\@ifnextchar[{\@prolpf@wrt@at{#1}{#2}}{\@prolpf@wrt@{#1}{#2}}}
\def\@prolpf@wrt@at#1#2[#3]{\prolsymbol^{#1}_{#3}#2}
\def\@prolpf@wrt@#1#2{\prolsymbol^{#1}#2}
\def\@prolpf@#1{\@ifnextchar[{\@prolpf@at{#1}}{\@prolpf@@{#1}}}
\def\@prolpf@at#1[#2]{\prolsymbol_{#2}#1}
\def\@prolpf@@#1{\prolsymbol#1}
\def\@proldf{\@ifnextchar[{\@proldf@wrt}{\@proldf@}}
\def\@proldf@wrt[#1]#2{\@ifnextchar[{\@proldf@wrt@at{#1}{#2}}{\@proldf@wrt@{#1}{#2}}}
\def\@proldf@wrt@at#1#2[#3]{\prolsymbol^{*#1}_{#3}#2}
\def\@proldf@wrt@#1#2{\prolsymbol^{*#1}#2}
\def\@proldf@#1{\@ifnextchar[{\@proldf@at{#1}}{\@proldf@@{#1}}}
\def\@proldf@at#1[#2]{\prolsymbol^*_{#2}#1}
\def\@proldf@@#1{\prolsymbol^*#1}
\def\prolsymbol{\mathcal{T}}
\makeatother

\usepackage{authblk}

\begin{document}
	
	\title{Metriplectic Euler-Poincar\'e equations: smooth and discrete dynamics}

  \author[1]{Anthony Bloch}
  \author[2]{ Marta Farr\'e Puiggal{\'\i}} 
  \author[3]{David Mart{\'\i}n de Diego}

\affil[1]{Department of Mathematics, University of Michigan 530 Church Street, Ann Arbor, MI, USA  }
\affil[3]{Instituto de Ciencias Matem\'aticas,  ICMAT (CSIC-UAM-UC3M-UCM)
	Madrid, Spain}

\maketitle

 \begin{abstract}In this paper we will study some interesting properties of modifications of the Euler-Poincar\'e equations when we add a special type of dissipative force, 
 so that the equations of motion can be described using the metriplectic formalism. The metriplectic   representation of the dynamics allows us to describe the   conservation of energy, as well as to guarantee entropy production.  Moreover, we describe   the use of discrete gradient systems to numerically simulate the evolution of the continuous metriplectic equations  preserving their main properties: preservation of  energy and correct entropy production rate. 
 \end{abstract}

\maketitle

 \section{Introduction}
	 In many examples of dynamics, especially in thermodynamics, it is necessary to combine  the dynamical structure of Hamiltonian systems and metric systems to produce what are called metriplectic systems, as originally discussed in the work of Morrison, see \cite{morrison,Bloch-Morrison-Ratiu} (see also \cite{KauMorr,Grmela-Ottinger}). The dynamics is determined using a Poisson bracket for the Hamiltonian part,  combined with a symmetric bracket which allows us to introduce dissipative effects.

 After introducing the notion of metriplectic system, in this paper we study  metriplectic systems derived from a perturbation of the Euler-Poincar\'e equations  or a Lie-Poisson  system by adding a special dissipation term \cite{bloch1996euler,esen}.
 Recall that the Euler-Poincar\'e equations are obtained by reduction from invariant Lagrangian  systems on the tangent bundle $TG$ of a Lie group $G$.  
 The dissipation term that we add to the equations  makes the equations of motion verify two interesting properties: preservation of energy $H$ and also the existence of a Casimir function $S$  of the Lie-Poisson bracket  verifying the property $\dot{S}\geq 0$. Both correspond exactly with the two laws of thermodynamics: preservation of the total energy and irreversible entropy creation. 
 
 To numerically approximate the solutions of a metriplectic system while preserving the energy and the entropy behaviour it is natural to use  a class of geometric integrators called discrete gradient methods. These methods are adequate when we want to preserve exactly the energy of the system. In this sense, they are quite useful for ODEs of the form $\dot{x} = \Pi(x)\nabla H(x)$ with   $x\in {\mathbb R}^n$ and  $\Pi(x)$ a skew-symmetric matrix (not necessarily associated to a Poisson bracket).
 Using a discrete gradient $\bar{\nabla} H(x,x')$ as an adequate approximation of the differential of the Hamiltonian function (see Section 4 for more details), it is possible to define a class of integrators
$x'-x=\bar{\Pi}(x,x')\bar{\nabla}H(x,x')$
 preserving the energy $H$ exactly,  i.e. $H(x)=H(x')$. Here $\bar{\Pi}(x,x')$ is a skew-symmetric matrix approximating $\Pi(x)$. In Section 4, based on discrete gradient methods, we derive geometric integrators  for metriplectic systems and in particular, the geometric derivation of the discrete dissipative term.

 \section{Metriplectic systems}
	
	The theory of {\bf metriplectic systems}  tries to combine together the dynamics generated by Poisson brackets with additional dissipative effects. 
	We will first describe  the different geometric elements  that define a metriplectic system.

	\subsection{ Poisson structures}

	Consider a differentiable manifold $P$ equipped with a Poisson structure  \cite{Abraham1978,marle} given by a bilinear map
	\[
	\begin{array}{ccc}
	C^{\infty}(P)\times C^{\infty}(P)&\longrightarrow& C^{\infty}(P)\\
	(f, g)&\longmapsto& \{f, g\}
	\end{array}
	\]
	called the {\bf Poisson bracket}, satisfying the following properties: 
	\begin{itemize}
		\item[(i)] \emph{Skew-symmetry},  $\{g, f\}=-\{f, g\}$;
		\item[(ii)] \emph{Leibniz rule}, $\{fg, h\}=f\{g, h\}+g\{f, h\}$; 
		\item[(iii)] \emph{Jacobi identity},  $\{\{f, g\}, h\}+\{\{h, f\}, g\}+\{\{g, h\}, f\}=0$;
	\end{itemize}
	for all $f, g, h\in C^{\infty}(P)$. 
	
	Given a Poisson manifold with bracket $\{\, ,\,\}$ and a function $f\in C^{\infty}(P)$ we may associate with $f$ a unique vector field $X_f\in {\mathfrak X}(P)$, the {\bf Hamiltonian vector field} defined by  
	$
	X_f(g)=\{g, f\}$.
	
	
	Moreover, on a Poisson manifold, there exists a unique bivector field $\Pi$, a Poisson bivector  (that is, a twice contravariant skew symmetric differentiable tensor) 
	such that 
	\begin{equation*}
		\{f,g\}:=\Pi(df ,d g), \qquad f, g\in C^\infty (P)\; .
	\end{equation*}
	The bivector field $\Pi$ is called the {\bf Poisson tensor} and the Poisson structure is usually denoted by $(P, \{\, , \}$) or $(P, \Pi)$. 
	The Jacobi identity in terms of the bivector $\Pi$ is written as
	$
	[\Pi, \Pi\,]=0,
	$
	where here $[\; , \; ]$  denotes the Schouten–Nijenhuis bracket (see \cite{Abraham1978}).

	Take coordinates  $(x^i)$, $1\leq i\leq \dim P=m$, and let $\Pi^{ij}$ be the components of the Poisson bivector, that is,
	\[
	\Pi^{ij}=\{x^i, x^j\}, \qquad \Pi=\Pi^{ij}\frac{\partial}{\partial x^i}\wedge \frac{\partial}{\partial x^j}\; . 
	\]
	Then if $f, g\in C^{\infty}(P)$ we have
	\[
	\{f, g\}=\sum_{i,j=1}^m\{x^i, x^j\}\frac{\partial f}{\partial x^i}\frac{\partial g}{\partial x^j}= \sum_{i,j=1}^m\Pi^{ij}\frac{\partial f}{\partial x^i}\frac{\partial g}{\partial x^j}\; ,
	\] 
 and the Hamiltonian vector field is written in local coordinates as
 \[
X_f=\Pi^{ij}\frac{\partial f}{\partial x^j}\frac{\partial }{\partial x^i}\; .
\]
	Observe that the $m\times m$ matrix $(\Pi^{ij})$  verifies the following properties:
	\begin{itemize}
		\item[(i)] \emph{Skew-symmetry}, $\Pi^{ij}=-\Pi^{ji}$  
		\item[(ii)] \emph{Jacobi identity}, 
		\[
		\sum_{l=1}^m\left(
		\Pi^{il}\frac{\partial \Pi^{jk}}{\partial x^l}+\Pi^{kl}\frac{\partial \Pi^{ij}}{\partial x^l}+\Pi^{jl}\frac{\partial \Pi^{ki}}{\partial x^l}\right)=0\, ,\quad  i, j, k=1,\ldots, m.
		\]
	\end{itemize}

	Define ${\sharp}_{\Pi} : T^{*}P \to TP$ by
	\[
	{\sharp}_\Pi (\alpha) = -\iota_\alpha \Pi=\Pi(\cdot, \alpha)  ,
	\]
	where $\alpha \in T^*P$,
	and $\langle \beta, \iota_\alpha \Pi\rangle=\Pi(\alpha, \beta)$ for all $\beta \in T^*P$.
	The rank of $\Pi$ at $p\in P$ is exactly the $\hbox{rank of }({\sharp}_\Pi)_p: T^*_pP\rightarrow T_pP$.  Because of the skew-symmetry of $\Pi$, we know that the rank of $\Pi$  at a point $p\in P$ is an even integer. 
	
	Given a function  $H \in
	C^\infty(P)$, a Hamiltonian function, we have the corresponding {\bf Hamiltonian vector field}:
	\[
	X_H = {\sharp}_\Pi (dH) .
	\]
	Therefore, on a Poisson manifold, a function $H$ determines the following  dynamical system:
	\begin{equation}\label{hamil-eq}
	\frac{dx}{dt}(t)=X_H(x(t))\; .
	\end{equation}
	Moreover, a function $f\in C^{\infty}(P)$  is a first integral of the Hamiltonian vector field  $X_H$ if for any solution 
	$x(t)$ of Equation (\ref{hamil-eq}) we have
	\[
	\frac{df}{dt}(x(t))=0\; .
	\]
	In other words, if  $X_H(f) = 0$ or, equivalently, $\{f, H\} = 0$. In particular,  the
	Hamiltonian function is a conserved quantity since $\{H, H\}=0$ by the skew-symmetry of the bracket.
	For any Poisson manifold $(P,\Pi)$  a function $C\in C^{\infty}(P)$
	is called a
	{\bf Casimir function} of $\Pi$ if $X_C = 0$, that is, if $\{C, g\} = 0$ for all $g\in C^{\infty}(P)$.

	\subsection{Positive semidefinite inner products}
	
	Assume that for  each point $x\in P$ we have a {\bf positive semidefinite inner product} for covectors 
	\[
	{\mathcal K}_x: T^*_xP\times T^*_xP \rightarrow {\mathbb R}
	\]
	from which we can define ${\sharp}_{{\mathcal K}} : T^{*}P \to TP$ by 
	\[
	{\sharp}_{\mathcal K} (\alpha_x)={\mathcal K}_x(\alpha_x, \cdot)  
	\]
	  and  the {\bf gradient vector field} 
	\[
	\hbox{grad}^{\mathcal K} S={\sharp}_{{\mathcal K}}(dS)
	\]
	for any function $S: P\rightarrow {\mathbb R}$. 
	
${\mathcal K}$ defines a {\bf symmetric bracket} given by
\[
(df, dg)={\mathcal K}(df, dg).
\]
	 
Take coordinates  $(x^i)$, $1\leq i\leq \dim P=m$, and let  $K^{ij}$ be the components of the inner product given by 
\[
{ K}^{ij}=(x^i, x^j).
\]
Then if $f, g\in C^{\infty}(P)$, the symmetric bracket is expressed as
\[
(f, g)=\sum_{i,j=1}^m(x^i, x^j)\frac{\partial f}{\partial x^i}\frac{\partial g}{\partial x^j}= \sum_{i,j=1}^m K^{ij}\frac{\partial f}{\partial x^i}\frac{\partial g}{\partial x^j}\; .
\] 
Observe that the $m\times m$ matrix $(K^{ij})$  verifies  $K^{ij}=K^{ji}$ and all the  eigenvalues are positive or zero.

\subsubsection{A construction of the positive semidefinite inner product with 	${\sharp}^{\mathcal K}(dH)=0$  given a Riemannian metric}
	
	 Assume that $P$ is equipped with a Riemannian metric ${\mathcal G}$ inducing a  positive definite inner product ${\mathcal G}^*$ on $T^*P$,
	\[
{\mathcal G}^*: T^*_xP\times T^*_xP	\rightarrow {\mathbb R}
\] 
	defined by ${{\mathcal G}}^*(df, dg)={\mathcal G}(\hbox{grad}f, \hbox{grad} g)$. 
	
 Since we are interested in defining a semidefinite inner product ${\mathcal K}$ such that ${\mathcal K}(dH, \cdot)=0$ then
we define
\begin{eqnarray}
{\mathcal K}(df, dg)&=&\frac{1}{{\mathcal G}^*(dH, dH)} {\mathcal G}^*({\mathcal G}^*(dH, dH)df-{\mathcal G}^*(dH, df)dH, {\mathcal G}^*(dH, dH)dg-{\mathcal G}^*(dH, dg)dH)\nonumber\\
&=&
{\mathcal G}^*(dH, dH){\mathcal G}^*(df, dg)-
{\mathcal G}^*(dH, df){\mathcal G}^*(dH, dg). \label{semid}
\end{eqnarray}
In coordinates we have
	\[
	K^{ij}=C_Hg^{ij}-g^{i\bar{j}}\frac{\partial H}{\partial x^{\bar{j}}}g^{\bar{i}j}\frac{\partial H}{\partial x^{\bar{i}}},
	\]		
	where $C_H=g^{ij}\frac{\partial H}{\partial x^{i}}\frac{\partial H}{\partial x^{j}}$, $(g_{ij})$ are the components of the Riemannian metric in a given coordinate system and 
 $(g^{ij})$ denotes its inverse matrix. 
	
	By construction ${\mathcal K}$ is positive semidefinite  and ${\mathcal K}(dH, \cdot)=0$.
	\begin{remark}
	Additionally we can add new functions $L_a: P\rightarrow {\mathbb R}$, $1\leq a\leq N$ to this construction in such a way that ${\mathcal K}(dL_a, \cdot)=0$, considering 
	\[
	{\mathcal K}(df, dg)=  {\mathcal G}^*(df-C^{ab}{\mathcal G}^*(dL_a, df) dL_b, dg-C^{ab}{\mathcal G}^*(dL_a, dg)dL_b)
	\]
	where $C_{ab}= {\mathcal G}^*(dL_a, dL_b)$, $1\leq a\leq N$ and $L_1=H$.
 \end{remark}
\subsection{Metriplectic systems}\label{section:metriplectic}
	
	A metriplectic system consists of a smooth manifold $P$, two smooth vector bundle maps 
 $\sharp_{\Pi}, \sharp_{{\mathcal K}}: T^*P\rightarrow TP$ covering the identity, and two functions $H, S\in C^{\infty}(P)$ called the  Hamiltonian (or total energy) and the entropy of the system, such that for all $f, g\in  C^{\infty}(P)$:
\begin{itemize}
 \item $\{f, g\}=\langle df, \sharp_{\Pi}(dg)\rangle$ is a Poisson bracket ($\Pi$ denotes the Poisson bivector).
 \item  $(f, g)=\langle df, \sharp_{{\mathcal K}}(dg)\rangle $ is a positive semidefinite symmetric bracket, i.e., $(\cdot,\cdot)$ is bilinear and symmetric.
 \item ${\sharp}_{\mathcal K}(dH)=0$    or equivalently   $(H, f)=0$,  $\forall f\in C^{\infty}(P)$.
 \item  ${\sharp}_{\Pi}(dS)=0$     or equivalently  $\{S, f\}=0$,  $\forall f\in C^{\infty}(P)$, that is,    $S$ is a Casimir function for the Poisson bracket. 
 \end{itemize}

	 Consider the function $E=H+S: P\rightarrow {\mathbb R}$.  Then the dynamics of the  metriplectic system is determined by
	\begin{align*}
		\frac{dx}{dt}&=\sharp_{\Pi}(dE(x(t)))+\sharp_{\mathcal K}(dE(x(t)))\\
		&= \sharp_{\Pi}(dH(x(t)))+\sharp_{{\mathcal K}}(dS(x(t)))\\
		&=X_H(x(t))+	\hbox{grad}^{\mathcal{K}} S(x(t)),
	\end{align*}
 where $X_H=\sharp_{\Pi}(dH)$ and $\hbox{grad}^{\mathcal K} S=\sharp_{{\mathcal K}}(dS)$. 
	From the equations of motion it is simple to deduce the following: 
	\begin{itemize}
		\item {\bf First law: conservation of energy}, $\frac{dH}{dt}=\{H, H\}+(H, S)=0$	
		\item {\bf Second law: entropy production}, $\frac{dS}{dt}=(S, S)\geq 0$. 
	\end{itemize}
	Thus, metriplectic dynamics embodies both the first and second laws of thermodynamics.

	In coordinates, the dynamics of the metriplectic system is written as
	\[
	\dot{x}^i=\Pi^{ij}\frac{\partial H}{\partial x^j}+K^{ij}\frac{\partial S}{\partial x^j}, \quad 1\leq i\leq n
	\]
	or, in matrix form, as
	\begin{equation}\label{matricial}
	\dot{x}=\Pi{ \nabla H}+K\nabla S, \quad 1\leq i\leq n.
	\end{equation}

		\subsection{Symmetry preservation}
	
	Let $\Phi: G\times P\rightarrow P$ be a smooth (left) action of a Lie group $G$ on $P$, given by $\Phi(g, x)=\Phi_g(x)=g\cdot x$ with $g\in G$ and $x\in P$. The action satisfies the following properties: 
 \begin{itemize}
\item $\Phi(e, x)=x$ where $e$ is the neutral element of $G$;
\item For every $g_1, g_2\in G$ and for every $x\in P$
\[
\Phi(g_1, \Phi(g_2, x))=\Phi(g_1g_2, x)\; .
\]
 \end{itemize}
	The infinitesimal generator of the action corresponding to a Lie algebra
	element $\xi \in {\mathfrak g}$ is the vector field $\xi_P$ on $P$ given by
	\[
	\xi_P(x)=\frac{d}{dt}\Big|_{t=0}(\hbox{exp}(\xi t)\cdot x).
	\]
Let $P$ be a Poisson manifold with Poisson bracket $\{\; ,\; \}$ and 
	assume that the action $\Phi$ is  a {\bf Poisson action}, that is, 
	\[
	\Phi_g^*\{f, h\}=\{\Phi_g^*f, \Phi_g^*h\}\; ,\quad \forall f, h\in C^{\infty}(P) \quad \forall g\in G\; .\]

 A {\bf momentum map} for the action $\Phi$ is a smooth map
$J: P\rightarrow {\mathfrak g}^*$  such that for each $\xi\in {\mathfrak g}$, the associated map
$J_{\xi}: P\rightarrow {\mathbb R}$ defined by 
$J_{\xi}(x)=\langle J(x), \xi\rangle $
satisfies that 
	$
	X_{J_{\xi}}=\xi_P$
	for all  $\xi\in {\mathfrak g}$ where 
 $X_{J_{\xi}}(f)=\{f, J_{\xi}\}$. As a consequence,  for any function $f\in C^{\infty}(P)$
	\[
	\{f, J_{\xi}\}=\xi_P(f).
	\]
	
	If the Lie algebra ${\mathfrak g}$  acts  on the Poisson manifold $P$  and admits a momentum map $J:  P\rightarrow {\mathfrak g}^*$, and if $H\circ \Phi_g=H$ (which is equivalent to $\xi_P(H)=0$ for all $\xi\in {\mathfrak g}$), then $J_{\xi}$ is a constant of the motion of $X_H$.

	Additionally, for the metriplectic system we will  assume that 
	\[
	(f, J_\xi)=0, \quad \forall \xi \in {\mathfrak g} \hbox{   and  }  f\in C^{\infty}(P),
	\]
	or equivalently $\sharp_{\mathcal K}(J_\xi)=0$. 
	Then for the metriplectic system we have
	\[
	\frac{dJ_\xi}{dt}=\{J_\xi, H \}+(J_\xi, S)=0
	\]
	and therefore $J_\xi: P\rightarrow {\mathbb R}$  is a constant of motion of the metriplectic system. 
	
	As a particular case of the previous construction we will consider in the next section the case when $P=TG$, where $G$ is  a Lie group, and we consider as a left action $\Psi_g=T^*\mathcal{L}_{g}: T^*G\rightarrow T^*G$, where $\mathcal{L}_{g}:G\rightarrow G$ is the left action. Under the symmetry conditions, the system reduces to a metriplectic system on ${\mathfrak g}^*$, the dual of the Lie algebra of $G$.

	\section{Forced  Euler-Poincar\'e equations and  metriplectic dynamics}
	Consider a Lagrangian system $l:{\mathfrak g}\rightarrow {\mathbb R}$, where ${\mathfrak g}$ is a Lie algebra, 
	and its corresponding Euler-Poincar\'e equations  \cite{holm1998euler,holm2009geometric}:
	\begin{equation}\label{el}
	\frac{d}{dt}\left(\frac{\delta l}{\delta \xi}\right)=ad_{\xi}^*	\frac{\delta l}{\delta \xi},
	\end{equation}
 where $\xi\in {\mathfrak g}$ and $\langle ad_{\xi}^*\alpha, \xi'\rangle=\langle \alpha, [\xi, \xi']\rangle$ for all $\xi'\in {\mathfrak g}$ and $\alpha\in {\mathfrak g}^*$.
	From this equation it is clear that the energy $E_l=\langle\frac{\delta l}{ \delta \xi}, \xi\rangle-l$ of the system is preserved, that is,  
	\[
		\frac{d E_l}{dt}=\frac{d}{dt}\left(  \langle \frac{\delta l}{ \delta \xi}, \xi\rangle-l\right)=0\; .
	\]
	However, there are other variations of this system that are subjected to external forces that also preserve energy. This class of systems is interesting in thermodynamics when we work with a closed system, as we have seen in the Subsection \ref{section:metriplectic} (see also \cite{morrison, bloch1996euler}). 
	For instance, if we add an external force $F: {\mathfrak g}\rightarrow {\mathfrak g}^*$ of the form
	\[
	F(\xi')=ad^*_{\xi} {\mathcal F}(\xi'), \qquad \xi'\in {\mathfrak g}
	\]
	where ${\mathcal F}: {\mathfrak g}\rightarrow {\mathfrak g}^*$ is an arbitrary map, then the {\bf forced Euler-Lagrange equations} are
		\begin{equation}\label{fel}
	\frac{d}{dt}\left(\frac{\delta l}{ \delta \xi}\right)=ad_{\xi}^*	\frac{\delta l}{\delta \xi} +F=ad_{\xi}^*	\left[\frac{\delta l}{ \delta \xi}+{\mathcal F}\right].
	\end{equation}
	
	Assume that ${\mathfrak g}$ is finite dimensional and  $\{e_a\}$, $1\leq a\leq n=\dim {\mathfrak g}$ is a basis of the Lie algebra with structure constants $C_{ab}^d$, that is,
	\[
	[e_a, e_b]=C_{ab}^d e_d,
	\]
	and denote by $(\xi^a(t))$ the coordinates of a curve  $\xi(t)\in {\mathfrak g}$. Then the equations (\ref{fel}) are
	 	\begin{equation}\label{fel2}
	 \frac{d}{dt}\left(\frac{\delta l}{ \delta \xi^b}(\xi(t))\right)=C_{ab}^d \xi^a(t) \left(\frac{\delta l}{ \delta \xi^d}(\xi(t))+{\mathcal F}_d(\xi(t))\right),
	 \end{equation}
	where ${\mathcal F}(\xi)={\mathcal F}_{d}(\xi) e^d$ and  $\{e^a\}$, $1\leq a\leq n$, is the dual basis of $\{e_a\}$.
\begin{example}	
  In the case of $G=SO(3)$ if we identify its Lie algebra ${\mathfrak g}$ with ${\mathbb R}^3$ with the usual vector cross product then we have
	 \[
	 	\frac{d}{dt}\left(\frac{\delta l}{ \delta {\mathbf \Omega}}\right)=\frac{\delta l}{ \delta \mathbf \Omega}\times {\mathbf  \Omega}+{\mathcal F}\times {\mathbf \Omega}
	 \]
	 as a generalization of the equations of the rigid body also preserving the total energy of the system. In particular if 
	 \[
	 l(\Omega_1, \Omega_2, \Omega_3)=\frac{1}{2}(I_1\Omega_1^2+I_2\Omega_2^2+I_3\Omega_3^2)
	 \]
	 then (\ref{fel}) are
	 	\begin{eqnarray*}
	 		I_1\dot{\Omega}_1&=&(I_2-I_3)\Omega_2\Omega_3+\Omega_3{\mathcal F}_2({\mathbf  \Omega})-\Omega_2{\mathcal F}_3({\mathbf  \Omega}),\\
	 			I_2\dot{\Omega}_2&=&(I_3-I_1)\Omega_3\Omega_1+\Omega_1{\mathcal F}_3({\mathbf  \Omega})-\Omega_3{\mathcal F}_1({\mathbf  \Omega}),\\
	 	I_3\dot{\Omega}_3&=&(I_1-I_2)\Omega_1\Omega_2+\Omega_2{\mathcal F}_1({\mathbf  \Omega})-\Omega_1{\mathcal F}_2({\mathbf  \Omega}),
	 \end{eqnarray*}
  where ${\mathcal F}=({\mathcal F}_1, {\mathcal F}_2, {\mathcal F}_3)$ and $\mathbf  \Omega=(\Omega_1, \Omega_2, \Omega_3)$.
\end{example}
 
 Using the Legendre transformation (that we assume in the sequel that is a local diffeomorphism)  we can
	write the forced Euler-Lagrange equations as
	\begin{equation}\label{flp}
	\dot{\mu}=ad^*_{\delta H/\delta \mu} \left(\mu+{\mathcal F}\left(\frac{\delta H}{\delta \mu}\right)\right),
	\end{equation}
 where $H$ is defined by $H(\mu)=\langle \mu, \xi(\mu)\rangle -L(\xi(\mu))$ and  $\mu=\frac{\delta l}{\delta \xi}(\xi)$. 
	This is a particular case of the {\bf forced Lie-Poisson equations} \cite{bloch1996euler}.

	Now, if $C: {\mathfrak g}^*\rightarrow {\mathbb R}$ is a Casimir function for the Lie-Poisson bracket of ${\mathfrak g}^*$ then along the evolution of the system (\ref{flp}) we have
	\begin{equation}\label{cas}
	\frac{dC}{dt}=\left\langle  {\mathcal F}\left(\frac{\delta H}{\delta \mu}\right), \left[
	\frac{\delta H}{\delta \mu}, 	\frac{\delta C}{\delta \mu}
	\right]
	\right\rangle .
	\end{equation}
 \begin{example}
	For the rigid body equations with Hamiltonian  and Casimir functions given by 
	\begin{eqnarray*}
		H(\Pi_1, \Pi_2, \Pi_3)&=&\frac{1}{2}\left(\frac{\Pi_1^2}{I_1}+\frac{\Pi_2^2}{I_2}+\frac{\Pi_3^2}{I_3}\right)\, ,\\
	     C(\Pi_1, \Pi_2, \Pi_3)&=&\frac{1}{2}(\Pi_1^2+\Pi_2^2+\Pi_3^2)\, ,
	\end{eqnarray*}
	in induced coordinates $(\Pi_1=I_1 \Omega_1, \Pi_2=I_2\Omega_2, \Pi_3=I_3\Omega_3)$ on $({\mathbb R}^3)^*\equiv {\mathbb R}^3$, 
	Equation (\ref{cas}) is
	\[
		\frac{dC}{dt}=\left(\frac{1}{I_2}-\frac{1}{I_3}\right)\Pi_2\Pi_3{\mathcal F}_1
	+	\left(\frac{1}{I_3}-\frac{1}{I_1}\right)\Pi_1\Pi_3{\mathcal F}_2
	+\left(\frac{1}{I_1}-\frac{1}{I_2}\right)\Pi_1\Pi_2{\mathcal F}_3.
	\]
For instance, if we take ${\mathcal F}: {\mathfrak g}\equiv {\mathbb R}^3\rightarrow {\mathfrak g}^*\equiv {\mathbb R}^3$ as 
	\[
	{\mathcal F}({\mathbf \Omega})=((I_3-I_2)\Omega_2\Omega_3, (I_1-I_3)\Omega_1\Omega_3, (I_2-I_1)\Omega_1\Omega_2) \, ,
	\]
then we get
	\[
	\frac{dC}{dt}\geq 0\, .
	\]
	As in the  case of metriplectic systems, we have a system verifying the first and second laws of thermodynamics:
		\begin{eqnarray*}\label{relaxed}
		\dot{\Pi}_1&=&\frac{(I_2-I_3)}{I_2I_3}\Pi_2\Pi_3+\frac{(I_1-I_3)}{I_1I_3^2}\Pi_1\Pi_3^2-\frac{(I_2-I_1)}{I_1I_2^2}\Pi_1\Pi^2_2\, ,\\
		\dot{\Pi}_2&=&\frac{(I_3-I_1)}{I_3I_1}\Pi_3\Pi_1+\frac{(I_2-I_1)}{I_1^2I_2}\Pi_1^2\Pi_2-
		\frac{(I_3-I_2)}{I_2I_3^2}\Pi_2\Pi_3^2\, ,\\
		\dot{\Pi}_3&=&\frac{(I_1-I_2)}{I_1I_2}\Pi_1\Pi_2+\frac{(I_3-I_2)}{I_2^2I_3}\Pi_2^2\Pi_3-
		\frac{(I_1-I_3)}{I_1^2I_3}\Pi^2_1\Pi_3\, .
	\end{eqnarray*}
	These are the equations of the relaxed rigid body \cite{morrison}. 
	\end{example}
	
	Motivated by this example,  we want to study the possible  families of functions ${\mathcal F}: {\mathfrak g}\rightarrow {\mathfrak g}^*$ such that 
	\[
	\frac{dC}{dt}\geq 0
	\]
	and then our systems will automatically  verify the second law of thermodynamics where the Casimir function $C$ will play the role of the entropy. 
	
	Given an arbitrary positive semidefinite   scalar product on ${\mathfrak g}$
	\begin{equation}\label{sem}
	{\mathcal K}: {\mathfrak g}\times {\mathfrak g}\longrightarrow {\mathbb R}
	\end{equation}
	we can define ${\mathcal F}$ by 
	\begin{equation}\label{ops1}
	\langle {\mathcal F}(\xi), \eta\rangle ={\mathcal K}( \eta, [\xi, \frac{\partial C}{\partial \mu}])
	\end{equation}
	for all $\eta\in {\mathfrak g}$. 
	
	With this definition it is obvious that
	\begin{eqnarray*}
		\frac{dC}{dt}&=&\left\langle  {\mathcal F}\left(\frac{\delta H}{\delta \mu}\right), \left[
	\frac{\delta H}{\delta \mu}, 	\frac{\delta C}{\delta \mu}
	\right]
	\right\rangle \\
	&=&{\mathcal K}\left( \left[
	\frac{\delta H}{\delta \mu}, 	\frac{\delta C}{\delta \mu}
	\right], \left[	\frac{\delta H}{\delta \mu}, \frac{\partial C}{\partial \mu}\right]\right)\geq 0.
	\end{eqnarray*}
	

	\qquad

	\section{Generic integrators}
	
	In this section we will derive a second order integrator preserving some of the properties of a metriplectic system. The methods are based on the discrete gradient methods that are typically used for systems defined by an almost-Poisson bracket and in this case, the methods guarantee the exact preservation of the energy.  We will start with the classical methods where $P={\mathbb R}^n$ and after this we will discuss the case of $P$ being a general differentiable manifold. 

	\subsection{ Discrete gradient systems}
	
	For ODEs in skew-gradient form, i.e. $\dot{x}=\Pi(x)\nabla H(x)$ with $x\in \mathbb{R}^n$ and $\Pi(x)$ a skew-symmetric matrix, it is immediate to check that $H$ is a first integral. Indeed 
	\[
	\dot{H}=\nabla H(x)^T\dot{x}=\nabla H(x)^T \Pi(x)\nabla H(x)=0\, ,
	\]
	due to the skew-symmetry of $\Pi$. Using discretizations of the gradient $\nabla H(x)$ it is possible to define a class of integrators which preserve the first integral $H$ exactly.

	\begin{definition}\label{def31}
		Let $H:\mathbb{R}^n\longrightarrow \mathbb{R}$ be a differentiable function. Then $\bar{\nabla}H:\mathbb{R}^{2n}\longrightarrow \mathbb{R}^n$ is a discrete gradient of $H$ if it is continuous and satisfies
		\begin{subequations}
			\label{discGrad}
			\begin{align}
				\bar{\nabla}H(x,x')^T(x'-x)&=H(x')-H(x)\, , \quad \, \mbox{ for all } x,x' \in\mathbb{R}^n  \, ,\label{discGradEn} \\
				\bar{\nabla}H(x,x)&=\nabla H(x)\, , \quad \quad \quad \quad \mbox{ for all } x \in\mathbb{R}^n  \, . \label{discGradCons}
			\end{align}
		\end{subequations}
	\end{definition}
	
	Some well-known examples of discrete gradients are:
	\begin{itemize}
		\item The mean value (or averaged) discrete gradient introduced in \cite{HLvL83} and given by
		\begin{equation}
		\label{AVF}
		\bar{\nabla}_{1}H(x,x'):=\int_0^1 \nabla H ((1-\xi)x+\xi x')d\xi \, , \quad \mbox{ for } x'\not=x \, .
		\end{equation}
		
		\item The midpoint (or Gonzalez) discrete gradient, introduced in \cite{GONZ} and given by
		\begin{align}
			\bar{\nabla}_{2}H(x,x')&:=\nabla H\left( \frac{1}{2}(x'+x)\right)+\frac{H(x')-H(x)-\nabla H\left( \frac{1}{2}(x'+x)\right)^T(x'-x)}{|x'-x|^2}(x'-x) \, , \label{gonzalez}\\
			& \mbox{ for } x'\not=x \, . \nonumber
		\end{align}
		\item The coordinate increment discrete gradient, introduced in \cite{ITOH}, with each component given by
		\begin{equation}
		\label{itoAbe}
		\bar{\nabla}_{3}H(x,x')_i:=\frac{H(x'_1,\ldots,x'_i,x_{i+1},\ldots,x_n)-H(x'_1,\ldots,x'_{i-1},x_{i},\ldots,x_n)}{x'_i-x_i}\, , \quad 1\leq i \leq n\, ,
		\end{equation}
		when $x_i'\not=x_i$, and $\bar{\nabla}_{3}H(x,x')_i=\frac{\partial H}{\partial x_i}(x'_1,\ldots,x'_{i-1},x'_i=x_{i},x_{i+1},\ldots,x_n)$ otherwise.
	\end{itemize}

\subsection{Construction of Metriplectic or Generic integrators}

The idea is to construct a geometric integrator preserving as much as possible the properties of the continuous metriplectic Euler-Poincar\'e equations and, in particular, preserving the two laws of thermodynamics. 
 We are in the category of generic integrators \cite{Ottinger+2018+89+100,Shang-Ottinger} since we will use a discretization of the differential of $H$ using a discrete gradient, and a discretization of the positive semidefinite  inner product ${\mathcal K}$. 

Consider a Gonzalez' discrete gradient $\bar{\nabla}_2H: {\mathbb R}^{2n}\rightarrow {\mathbb R}^n$, the Poisson tensor $\Pi(z)$ where $z=\frac{x+x'}{2}$, 
and a discretization ${\mathcal K}_d$ of the inner product ${\mathcal K}$ which is also positive semidefinite. Then the {\bf generic integrator} is constructed as a discretization of equation (\ref{matricial}) as follows: 
\begin{equation}\label{ops}
\frac{x'-x}{h}=\Pi(z)\bar{\nabla}_2H (x, x')+K_d(z)\nabla S(z)
\end{equation}

For any $x\in P$,  we assume that the numerical scheme  (\ref{ops}) generates a local evolution in a neighborhood $U$ of $x$, 
in the sense that there exist  real numbers $\bar{h}, T >0$, and a discrete flow map  $\varphi: U\times [0,\bar{h}]\rightarrow P$ such that  for any $x_0\in U$ and $h\in [0, \bar{h}]$ the sequence $\{x_k\}$ generated by 
\[
x_k=\varphi(x_{k-1}, h)=\varphi^k(x_0, h)
\]
satisfies equation (\ref{ops}) for all $k$ such that $kh\in [0, T]$.

\begin{proposition}{[Second law]}
	The metriplectic integrator  verifies 
	\[
S(x_{k+1})-S(x_k)\geq {\mathcal O}(h^3) \quad \hbox{where}\quad x_{k+1}=\varphi(x_{k}, h) . 
	\]
\end{proposition}
	 \begin{proof}
	 Using Taylor's expansion we have that
	\begin{eqnarray*}
	&&S(x_{k+1})-S(x_k)+{\mathcal O}(|x_{k+1}-x_k|^3)= \nabla S(x_{k+1/2})^T(x_{k+1}-x_k)\\
	&&= h\nabla S(x_{k+1/2})^T\Pi(x_{k+1/2})\bar{\nabla}_2H (x_k, x_{k+1})+ h\nabla S(x_{k+1/2})^TK_d(x_{k+1/2})\nabla S(x_{k+1/2})
\geq 0\; ,
	\end{eqnarray*}
where $x_{k+1/2}=\frac{x_k+x_{k+1}}{2}$.
   \end{proof}
 
Now, for the exact preservation of the energy it is necessary to construct a discretization ${\mathcal K}_d$  of ${\mathcal K}$ given in (\ref{semid}) such that $\bar{\nabla}_2H$ is an element of the kernel of ${\mathcal K}_d$.

As in  (\ref{semid}) we  consider 
\begin{eqnarray}
{\mathcal K}_d(df, dg)&=& 
{\mathcal G}^*(\bar{\nabla}_2H, \bar{\nabla}_2H)\left[
{\mathcal G}^*(\bar{\nabla}_2H, \bar{\nabla}_2H){\mathcal G}^*(df, dg)-
{\mathcal G}^*(\bar{\nabla}_2H, df){\mathcal G}^*(\bar{\nabla}_2H, dg)\right]\label{semid-d}.
\end{eqnarray}

With the semi-definite positive inner product (\ref{semid-d}) we deduce the following.
\begin{proposition}{[First law]}
	The metriplectic integrator preserves exactly the energy function $H$, that is,
	\[
	H(x_{k+1})-H(x_k)=0.
	\]
\end{proposition}
\begin{proof}
	
	\begin{eqnarray*}
		H(x_{k+1})-H(x_k)&=& \bar{\nabla}_2H^T(x_k, x_{k+1})(x_{k+1}-x_k)\\
		&=& h\bar{\nabla}_2H^T(x_k, x_{k+1})\Pi(x_{k+1/2})\bar{\nabla}_2H (x_k, x_{k+1})\\
		&&+ h\bar{\nabla}_2H^T(x_k, x_{k+1})K_d(x_{k+1/2})\nabla S(x_{k+1/2})= 0
	\end{eqnarray*}
 since $K_d(x_{k+1/2})\bar{\nabla}_2H(x_k, x_{k+1})=0$.
\end{proof}

\subsection{Example: Numerical integration of the relaxing rigid body}

The rigid body equations are given by
\begin{eqnarray*}
	I_1\dot{\Omega}_1&=&(I_2-I_3)\Omega_2\Omega_3,\\ 
I_2\dot{\Omega}_2&=&(I_3-I_1)\Omega_1\Omega_3,\\ 
I_3\dot{\Omega}_3&=&(I_1-I_2)\Omega_1\Omega_2\; .
\end{eqnarray*}
These equations are the Euler-Poincar\'e equations for the Lagrangian $l: {\mathfrak g}\rightarrow {\mathbb R}$
\[
l(\Omega_1, \Omega_2, \Omega_3)=\frac{1}{2}I_1\Omega_1^2+I_2\Omega_2^2+I_3\Omega_3^2\; .
\]
Now, using the Legendre transformation we define the associated momenta: 
\[
p_1=\frac{\partial l}{\partial \Omega_1}=I_1\Omega_1, \quad p_2=\frac{\partial l}{\partial \Omega_2}=I_2\Omega_2,\quad p_3=\frac{\partial l}{\partial \Omega_3}=I_3\Omega_3\, .
\]	
Then the equations of motion of the system become
\begin{eqnarray*}
	\dot{p}_1&=&\frac{I_2-I_3}{I_2I_3}p_2 p_3,\\ 
	\dot{p}_2&=&\frac{I_3-I_1}{I_1I_3}p_1 p_3,\\ 
	\dot{p}_3&=&\frac{I_1-I_2}{I_1I_2}p_1 p_2.
\end{eqnarray*}
	This is a Lie-Poisson system and the equations are written in matrix form as
	\[
	\dot{p}=\Pi \nabla H=\sharp_{\Pi}(dH),
	\]
	where
	\[
	\Pi=\left(
	\begin{array}{ccc}
	0&-I_3p_3&I_2p_2\\
	I_3p_3&0&-I_1p_1\\
	-I_2p_2&I_1p_1&0
	\end{array}
	\right)
\]
	and $H(p_1, p_2, p_3)=\frac{1}{2}\left( \frac{p_1^2}{I_1}+\frac{p_2^2}{I_2}+\frac{p_3^2}{I_3}\right)$. 
		Consider now the positive semidefinite inner product defined by 
	\[
	{\mathcal K}(df, dg)=\left[
	\tilde{\mathcal G}(dH, dH)\tilde{\mathcal G}(df, dg)-
	\tilde{\mathcal G}(dH, df)\tilde{\mathcal G}(dH, dg)\right].
	\]
	After some straightforward computations using the canonical metric  of ${\mathbb R}^3$, we derive 
	that ${\mathcal K}$ is defined by the matrix
	\[
	K=\left(
	\begin{array}{ccc}
	\frac{p_2^2}{I_2^2}+\frac{p_3^2}{I_3^2}&-\frac{p_1 p_2}{I_1 I_2}&-\frac{p_1p_3}{I_1I_3}\\
-\frac{p_1 p_2}{I_1 I_2}&	\frac{p_1^2}{I_1^2}+\frac{p_3^2}{I_3^2}&-\frac{p_2 p_3}{I_2 I_3}\\
	-\frac{p_1 p_3}{I_1 I_3}&-\frac{p_2 p_3}{I_2 I_3}&	\frac{p_1^2}{I_1^2}+\frac{p_2^2}{I_2^2}
	\end{array}
	\right).
	\]
The entropy is defined by the Casimir function
\[
S(p_1, p_2, p_3)=\frac{1}{2}(p_1^2+p_2^2+p_3^2)
\]
and the dynamics of the metriplectic system is given by 
\[
\dot{p}=\Pi \nabla H +K\nabla S\; ,
\]
or
\begin{eqnarray*}
	\dot{p}_1&=&\frac{I_2-I_3}{I_2I_3}p_2 p_3+\left( \frac{1}{I_2^2}-\frac{1}{I_1I_2}\right)p_1p_2^2+\left( \frac{1}{I_3^2}-\frac{1}{I_1I_3}\right)p_1p_3^2,\\ 
	\dot{p}_2&=&\frac{I_3-I_1}{I_1I_3}p_1 p_3+\left( \frac{1}{I_1^2}-\frac{1}{I_1I_2}\right)p_2p_1^2+\left( \frac{1}{I_3^2}-\frac{1}{I_2I_3}\right)p_2p_3^2,\\ 
	\dot{p}_3&=&\frac{I_1-I_2}{I_1I_2}p_1 p_2+\left( \frac{1}{I_1^2}-\frac{1}{I_1I_3}\right)p_3p_1^2+\left( \frac{1}{I_2^2}-\frac{1}{I_2I_3}\right)p^2_2p_3.
\end{eqnarray*}
From construction we get $\Pi\nabla S=0$ and $K\nabla H=0$.

Using the notation $(P_1, P_2, P_3)=\varphi(p_1, p_2, p_3, h)$, the generic integrator is constructed taking
\begin{eqnarray*}
\bar{\nabla}_2 H(\frac{P_1+p_1}{2},\frac{P_2+p_2}{2}, \frac{P_3+p_3}{2})&=&\left(\frac{P_1+p_1}{2I_1},\frac{P_2+p_2}{2I_2}, \frac{P_3+p_3}{2I_3}\right)\\
&=&(z_1/I_1, z_2/I_2, z_3/I_3)
\end{eqnarray*}
and the discrete semidefinite scalar product 
	\[
K_d=\left(
\begin{array}{ccc}
\frac{z_2^2}{I_2^2}+\frac{z_3^2}{I_3^2}&-\frac{z_1 z_2}{I_1 I_2}&-\frac{z_1z_3}{I_1I_3}\\
-\frac{z_1 z_2}{I_1 I_2}&	\frac{z_1^2}{I_1^2}+\frac{z_3^2}{I_3^2}&-\frac{z_2 z_3}{I_2 I_3}\\
-\frac{z_1 z_3}{I_1 I_3}&-\frac{z_2 z_3}{I_2 I_3}&	\frac{z_1^2}{I_1^2}+\frac{z_2^2}{I_2^2}
\end{array}
\right).
\]
The metriplectic integrator is given in this case by the midpoint rule: 
\begin{eqnarray*}
\frac{P_1-p_1}{h}&=&\frac{I_2-I_3}{I_2I_3}z_2 z_3+\left( \frac{1}{I_2^2}-\frac{1}{I_1I_2}\right)z_1z_2^2+\left( \frac{1}{I_3^2}-\frac{1}{I_1I_3}\right)z_1z_3^2,\\ 
\frac{P_2-p_2}{h}&=&\frac{I_3-I_1}{I_1I_3}z_1 z_3+\left( \frac{1}{I_1^2}-\frac{1}{I_1I_2}\right)z_2z_1^2+\left( \frac{1}{I_3^2}-\frac{1}{I_2I_3}\right)z_2z_3^2,\\ 
\frac{P_3-p_3}{h}&=&\frac{I_1-I_2}{I_1I_2}z_1 z_2+\left( \frac{1}{I_1^2}-\frac{1}{I_1I_3}\right)z_3z_1^2+\left( \frac{1}{I_2^2}-\frac{1}{I_2I_3}\right)z_2^2z_3.
\end{eqnarray*}
In this case, since the Casimir is quadratic we have also that
\[
\nabla S=\bar{\nabla}_2 S,
\]
the system verifies that 
\[
S(P_1, P_2, P_3)-S(p_1, p_2, p_3)\geq 0\; ,
\]
and also $H(P_1, P_2, P_3)=H(p_1, p_2, p_3)$ for the discrete  flow $\varphi(p_1, p_2, p_3, h)=(P_1, P_2, P_3 )$. Then in this case the midpoint method preserves the energy exactly and moreover the entropy production rate has the correct behaviour.

\begin{figure}
    \centering
    \begin{subfigure}{0.3\textwidth}
        \centering
        \includegraphics[scale=0.3]{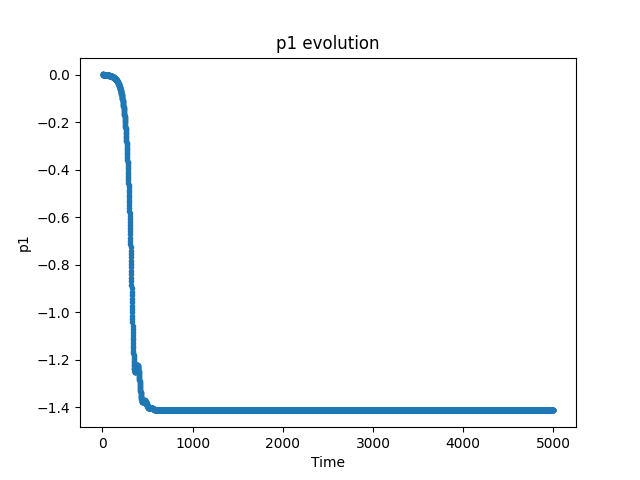}
    \end{subfigure}
    \begin{subfigure}{0.3\textwidth}
        \centering
        \includegraphics[scale=0.3]{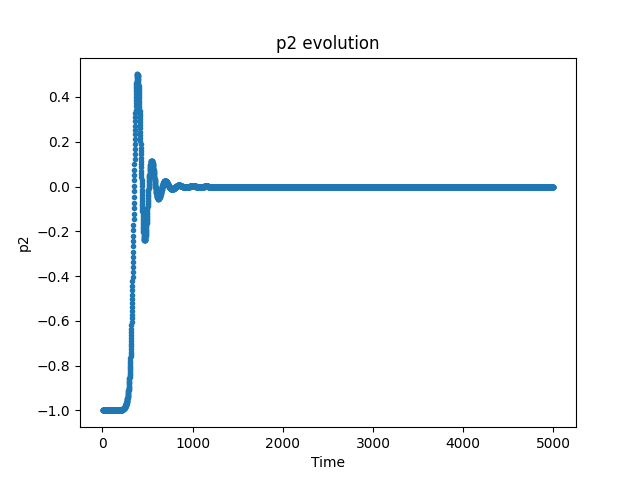}
    \end{subfigure}
    \begin{subfigure}{0.3\textwidth}
        \centering
        \includegraphics[scale=0.3]{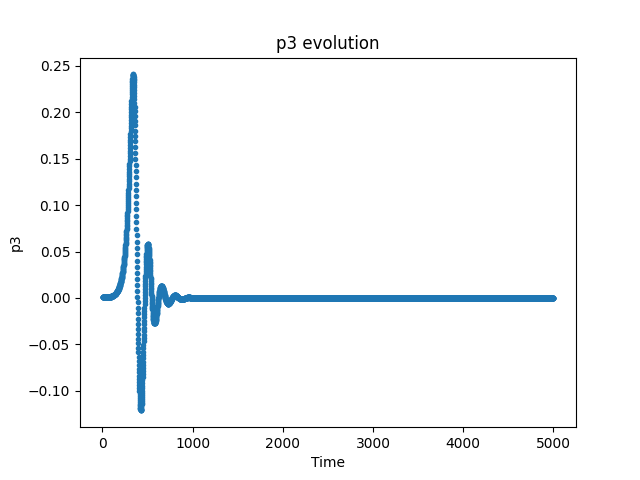}
    \end{subfigure}
    \medskip
    \begin{subfigure}{0.3\textwidth}
        \centering
        \includegraphics[scale=0.3]{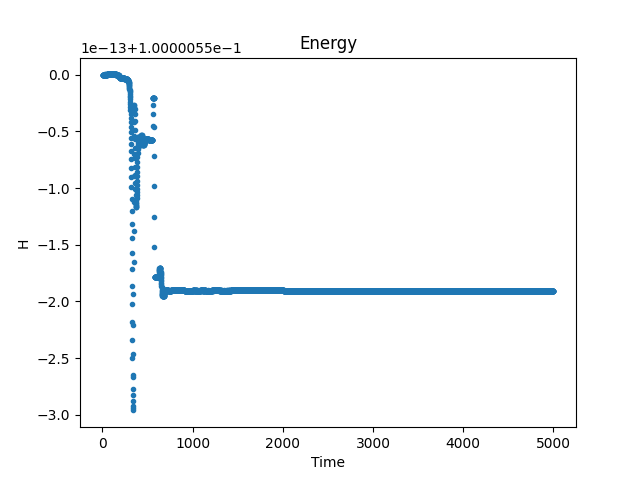}
    \end{subfigure}
    \begin{subfigure}{0.3\textwidth}
        \centering
        \includegraphics[scale=0.3]{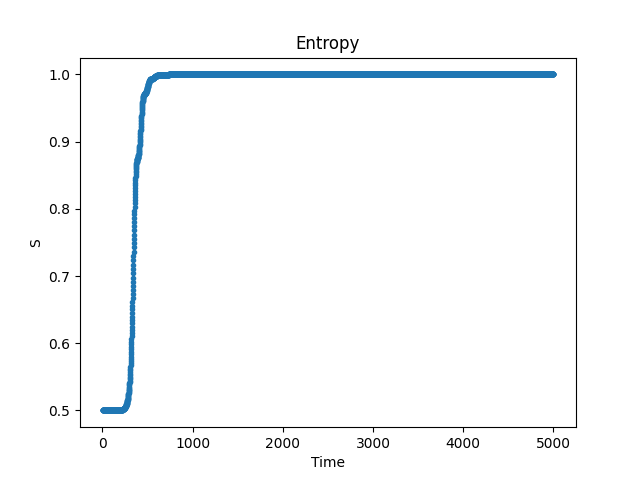}
    \end{subfigure}
    \caption{$I_1=10, I_2=5, I_3=1, h=0.1$, initial conditions $p_1=0.001, p_2=-1, p_3=0.001$}
\end{figure}

\subsection{Extension to differentiable manifolds}

We can extend this construction to the case where we are working on $P$  a general manifold.
	First we  will need to introduce 
	a finite difference map  or retraction map $R_h: U\subset TP\rightarrow P\times P$ and its inverse map  $R^{-1}_h: \bar{U}\subset P\times P\rightarrow TP$ \cite{absil}. For any $(x, x')\in \bar{U}$ we denote by $z=\tau_P(R_h^{-1}(x, x'))\in TP$.  
	We can use a type of retraction that is constructed using  an auxiliary  Riemannian metric ${\mathcal G}$ on $P$ with associated geodesic spray 
	$\Gamma_{\mathcal G}$ \cite{Celle}. The associated Riemannian exponential  for a  small enough $h>0$ is
constructed as 
	\[
	exp_{h}(v)=(\tau_Q(v), exp_{\tau_Q(v)}(hv)),
	\]
	where we have the standard exponential map on a Riemannian manifold defined by
	\[
	exp_{\tau_Q(v)}(v)=\gamma_v(1),
	\]
	where $t\rightarrow \gamma_v(t)$ is the unique geodesic  such that $\gamma_v'(0)=v$.
	Another interesting possibility related to the midpoint rule is 
	\begin{equation}\label{weyl}
	\widetilde{exp}_{h}(v)=(exp_{\tau_Q(v)}(-hv/2), exp_{\tau_Q(v))}(hv/2)).
	\end{equation}
	Both maps are local diffeomorphisms and then we can consider the corresponding inverse maps that we generically denote by $R^{-1}_h$ as at  the beginning of this section.

	Define a discrete gradient as a map
	$\bar{\nabla}{ H}: \bar{U}\subseteq P\times P \longrightarrow T^{*}P$ such that the following diagram commutes
	$$
	\xymatrix{
\bar{U}\subseteq 	P\times P \ar[rr]^{\bar{\nabla}{ H}} \ar[d]^{R^{-1}_h} && T^{*}P \ar[d]^{\pi_{P}}\\
		TP  \ar[rr]^{\tau_{P}} && P
	}
	$$
	and 
	 verifies the following two properties: 
	\begin{subequations}
		\label{discGrad-m}
		\begin{align}
			\langle \bar{\nabla} H(x,x'), R^{-1}_h(x,x')\rangle&=H(x')-H(x)\, , \quad \, \mbox{ for all } (x,x')\in \bar{U}   \, ,\label{discGradEn-m} \\
			\bar{\nabla}H(x,x)&= dH(x)\, , \quad \quad \quad \quad \mbox{ for all } x\in P  \, . \label{discGradCons-m}
		\end{align}
	\end{subequations}

In the case when we have a Riemannian metric ${\mathcal G}$ on $P$  we construct the following midpoint discrete gradient
\begin{align}
\bar{\nabla}_{2} H(x,x')&:=d H(z)+\frac{H(x')-H(x)-d H(z)(R^{-1}_h(x,x'))}{{\mathcal G}(R^{-1}_h(x, x'), R^{-1}_h(x, x'))}\flat_{\mathcal G}(R^{-1}_h(x, x')) \, , \label{gonzalez-m}\\
& \mbox{ for } x'\not=x \, , \nonumber
\end{align}
where 
$
\flat_{\mathcal G}:TP\rightarrow T^*P$ is given by 
$\flat_{\mathcal G}(u)(v)={\mathcal G}(u, v)$ for $u,v\in TP$ and $z=\tau_P(R^{-1}_h(x, x'))\in P$.

The metriplectic integrator that we propose  is written as
\[
R^{-1}_h(x_k, x_{k+1})=\Pi(z_\tau)\bar{\nabla}_{2}H(x_k,x_{k+1})+
K_d(z)\nabla C(z)
\] 
where $z=\tau_P(R^{-1}_h(x_k, x_{k+1}))$ and ${\mathcal K}_d$ is constructed as in (\ref{semid-d}).

\section*{Acknowledgements}
 A.B. was partially supported by NSF grant  DMS-2103026, and AFOSR grants FA
9550-22-1-0215 and FA 9550-23-1-0400.
 DMdD  acknowledges financial support from the Spanish Ministry of Science and Innovation, under grants  PID2022-137909NB-C21, TED2021-129455B-I00 and  CEX2019-000904-S funded by MCIN/AEI\-/10.13039\-/501100011033 and   BBVA Foundation via the project “Mathematical optimization for a more efficient, safer and decarbonized maritime transport”.

\bibliography{references-3.bib}
\bibliographystyle{unsrt}

\end{document}